\documentclass[a4paper]{amsart}

\usepackage{amsmath,amssymb,enumerate}
\usepackage{graphicx}
\usepackage{a4wide}
\usepackage{hyperref}

\newtheorem{theorem}{Theorem}
\newtheorem{lemma}{Lemma}
\newtheorem{proposition}{Proposition}
\newtheorem{corollary}{Corollary}

\newcommand{\doi}[1]{\href{http://dx.doi.org/#1}{\texttt{doi:#1}}}

\begin{document}

\title{Neighbor connectivity of $K$-ary $N$-cubes}

\author{Tom\'a\v{s} Dvo\v{r}\'ak \and Mei-Mei Gu}
\address{Faculty of Mathematics and Physics, Charles University, Prague, Czech Republic}
\email{dvorak@ksvi.mff.cuni.cz, mei@kam.mff.cuni.cz}

\begin{abstract}
The neighbor connectivity of a graph $G$ is the least number of vertices such that removing their closed neighborhoods from $G$ results in a graph that is  disconnected, complete or empty. If a~graph is used to model the topology of an interconnection network, this means that the failure of a network node causes failures of all its neighbors. 
We completely determine the neighbor connectivity of $k$-ary $n$-cubes for all $n\ge1$ and $k\ge2$.  
\end{abstract}



\maketitle

\thispagestyle{empty}

\section{Introduction}

The design and analysis of interconnection networks, originally incited by applications in telecommunication and computer networks, 
has become quite a~pervasive topic in both theoretical and applied research.
Among the problems arising in the course of network design, special attention has been paid to the aspect of fault-tolerance: Can the network preserve its functionality even if certain nodes become overloaded or faulty?

The topology of an interconnection network is usually modeled by an undirected graph whose vertices and edges represent network nodes  and communication links between them, respectively. As the simplest measure of fault-tolerance we can  use the graph-theoretic concept of {\em connectivity} $\kappa(G)$ of the underlying graph $G$: the least number of vertices  of $G$ whose removal leaves the resulting graph disconnected or trivial. There are, however, further refinements to this concept that better grasp the vulnerability or robustness of the network.

The \emph{neighbor connectivity} of a graph $G$, denoted by $\kappa_{NB}(G)$, is defined as the least number of vertices such that removing their closed neighborhoods from $G$ results in a graph that is  disconnected, complete or empty. If a~graph is used to model an interconnection network, this means that the failure of a network node causes failures of all its neighbors. The concept was introduced by Gunther and Hartnell  \cite{Gu} in the context of modeling spy networks; with Nowakowski \cite{G2} they proved that $\kappa_{NB}(G)\le\kappa(G)$ for every graph $G$. Doty  showed
that the problem to decide whether $\kappa_{NB}(G)\le k$ for a~given graph $G$ and integer $k$ is  NP-complete \cite{D2}. She also characterized Cayley graphs of neighbor connectivity one in terms of algebraic properties of the generating set \cite{D2}, sharpened the upper bound for abelian Cayley graphs and  formulated an open problem \cite{D1}: Is it true that $\kappa_{NB}(G)\le\lceil\frac{\delta}2\rceil$ for all $\delta$-regular abelian Cayley graphs except cycles? 
Our main result shows that this bound is tight for a class of abelian Cayley graphs formed by the $k$-ary $n$-cubes.

The concept of neighbor connectivity  was generalized to its edge version by Cozzens and Wu \cite{CoWu}: the minimum number of edges such that the removal of their closed neighborhoods results in an empty, trivial, or disconnected graph is called the \emph{edge neighbor connectivity}, denoted by $\lambda_{NB}(G)$.
Up to now, the exact values of $\kappa_{NB}(G)$  or $\lambda_{NB}(G)$ are known only for a~few classes of graphs.
Shang et al. \cite{S2} determined the (edge) neighbor connectivity of alternating group networks $AN_n$, star graphs $S_n$ and Cayley graphs generated by transposition trees $\Gamma_n$, showing that $\kappa_{NB}(AN_n) = n-1$ for $n\ge4$ while $\lambda_{NB}(AN_n) = n-2$ and $\kappa_{NB}(S_n) = \lambda_{NB}(\Gamma_n) = n -1$ for $n\ge3$. Zhao et al \cite{Z1} obtained the exact values of edge neighbor connectivity for the  class of  hypercubes $Q_n$ and  Cartesian product of complete graphs $K_n\, \square\, K_2$, namely $\lambda_{NB}(Q_{n})=n-1$ and $\lambda_{NB}(K_n\,\square\, K_2)=\lceil\frac{n}2\rceil$  for $n\geq 3$. 
Gunther and Hartnell  \cite{GH91} showed  that $\kappa_{NB}(K_{n}\,\square\, K_{n})=n-1$. 

In this paper we focus on a class of Cayley graphs formed by the $k$-ary $n$-cubes. Properties such as vertex-transitivity or hierarchical structure predetermine these graphs to serve as a~prospective candidate for the interconnection network design, and they have been widely studied from this point of view, see e.g. \cite{D97,G14,H12,Liu19}. Here we add another piece to the puzzle and completely determine the neighbor connectivity of $k$-ary $n$-cubes for all $n\ge1$ and $k\ge2$. 

The rest of the paper is laid out as follows. The next section introduces notation, terminology and surveys auxiliary results that shall be needed later. Section~\ref{sec:hypercubes} provides lower bounds on the neighbor connectivity of hypercubes while Section~\ref{sec:general} settles the general case of $k$-ary $n$-cubes with $k\ge3$. Finally, Section~\ref{sec:main} completes the picture by adding an upper bound applicable to all abelian Cayley graphs whose generating set satisfies certain minimal condition. The previous results are then summarized  into the main theorem showing that the upper bound is tight for the class of $k$-ary $n$-cubes, up to some trivial exceptions. The paper is concluded with brief suggestions for future research.  
\section{Preliminaries}

In this section we introduce the terminology and notations used throughout the paper. The graph-theoretic concepts undefined below may be found e.g.~in \cite{B}.
In the rest of this text, $n$ always denotes a positive integer while $[n]$ stands for the set $\{0,1,\dots, n-1\}$. 
As usual, we use $V(G)$ and $E(G)$ for the vertex and edge sets of a graph $G$. Given a set $S\subseteq V(G)$, an \emph{open neighborhood} $N(S)$ of $S$ is the set of all neighbors of vertices of $S$ in $G$, a \emph{closed neighborhood} $N[S]$ of $S$ is the set $S\cup N(S)$, and $G-S$ stands for the subgraph of $G$ induced by $V(G)\setminus S$. 
If $G$ is a regular graph, we use $\delta(G)$ to denote the degree of vertices of $G$.

\subsection{Paths}

A~sequence $(x_1$,$x_2$,$\dots,x_{n})$ of pairwise distinct vertices such that any two consecutive vertices are adjacent is a~\emph{path} of length $n-1$ with \emph{endvertices} $x_1$ and  $x_n$, also called an~\emph{$(x_1,x_n)$-path}.
We denote the vertex set $\{x_1,x_2,\dots,x_{n}\}$ of such a path $P$ by $V(P)$. 
Note that in the case that $n=1$, $P$ is just a~path of length $0$ consisting of a single vertex.  
We say that $(x,y)$-paths $\{P_i\}^n_{i=1}$ are \emph{internally disjoint} if $V(P_i)\cap V(P_j)=\{x,y\}$ for all $1\leq i<j\leq n$.
Paths $(x_1$,$x_2$,$\dots,x_{i})$ and $(x_{i}$,$x_{i+1}$,$\dots,x_{n})$,$1\le i \le n$, are called a \emph{prefix} and a \emph{suffix} of the path $(x_1$,$x_2$,$\dots,x_{n})$.

Given sets $X,Y\subseteq V(G)$, a path whose endvertices belong to $X$ and $Y$, respectively, is called an \emph{$(X,Y)$-path}.
If $X=\{x\}$ or $Y=\{y\}$, we simply speak about an $(x,Y)$-path or $(X,y)$-path. 
Given $x\in V(G)$ and $Y\subseteq V(G)$, an \emph{$(x,Y)$-fan}  (or \emph{$(Y,x)$-fan})  is a set of $|Y|$ $(x,Y)$-paths (or $(Y,x)$-paths) any two of which have exactly the vertex $x$ in common.

\subsection{Connectivity}

Recall that the {\em connectivity} $\kappa(G)$ of a graph $G$ is the least number of vertices  of $G$ whose removal leaves the resulting graph disconnected or trivial  (i.e., consisting of a single vertex). $G$ is called \emph{$k$-connected} if $\kappa(G)\ge k$. Menger's Theorem states that a~nontrivial graph $G$ is $k$-connected iff for any two distinct vertices $x,y\in V(G)$ there are at least $k$ internally disjoint $(x,y)$-paths.

The following simple property of $k$-connected graphs may be found e.g. in \cite[Lemma~9.3]{B}. 
\begin{proposition}
\label{lem:k-connected}
If $G$ is a $k$-connected graph, then the graph obtained from $G$ by adding a new vertex and joining it to at least $k$ vertices of $G$ is also $k$-connected. 
\end{proposition}

The next lemma recalls two useful consequences of Menger's Theorem (\cite[Propositions~9.4-9.5]{B}, the latter known as the Fan Lemma). As we need a slightly extended version for graphs with faulty vertices, both statements are provided with proofs.

\begin{lemma}
\label{lem:fan}
Let $f\ge0$ and $k\ge1$ be integers, $G$ be an $(f+k)$-connected graph, $F \subseteq V(G)$  with $|F|=f$ and $Y\subseteq V(G-F)$ with $|Y|=k$. Then
\begin{enumerate}[\upshape(i)]
	\item\label{lem:fan:xypaths} 
for any $X\subseteq V(G-F)$ such that $|X|=k$, there is a family of $k$ pairwise disjoint $(X,Y)$-paths in $G-F$;
	\item\label{lem:fan:fan} 
for any $x\in V(G-F)$ there is an $(x,Y)$-fan in $G-F$.
\end{enumerate}
 \end{lemma}
\begin{proof}
To prove part \eqref{lem:fan:fan}, obtain a new graph $H$ from $G$ by adding a new vertex $y$ and joining it to each vertex of $F\cup Y$. By Proposition~\ref{lem:k-connected}, $H$ is also $(f+k)$-connected. Therefore, by Menger's Theorem, there are $f+k$ internally disjoint $(x,y)$-paths in $H$. Deleting $y$ from each of these paths, we obtain $(f+k)$ internally disjoint $(x,F\cup Y)$-paths. It remains to consider only the $(x,Y)$-paths, observe that no vertex in any of them falls into $F$, and conclude that they form the desired $(x,Y)$-fan in $G-F$.

Part \eqref{lem:fan:xypaths} is now an easy corollary of \eqref{lem:fan:fan}: Indeed, let $H$ be a new graph obtained from $G$ by joining a new vertex $x$ to each vertex of $X$. Then $H$ is $(f+k)$-connected by Proposition~\ref{lem:k-connected} and therefore, by part \eqref{lem:fan:fan}, there is an $(x,Y)$-fan in $H-F$ formed by $k$ internally disjoint $(x,Y)$-paths. Deleting $x$ from each of these paths, we obtain the desired family of $k$ pairwise disjoint $(X,Y)$-paths in $G-F$. 
\end{proof}

Given a graph $G$ and a set $U\subseteq V(G)$ of \emph{faults}, the \emph{survival subgraph} of $G$ for $U$, denoted by $G\ominus U$, is the subgraph of $G$ induced by the set $V(G)\setminus N[U]$. 
In this context, we refer to the vertices of $N[U]$ as \emph{faulty} while the vertices of $G\ominus U$ are called \emph{healthy}. If $U=\{u\}$ we simply write $G\ominus u$.

\emph{Neighbor connectivity} of a graph $G$, denoted by $\kappa_{NB}(G)$, is defined as the minimum cardinality of a set $U\subseteq V(G)$ such that the survival subgraph $G\ominus U$ is disconnected, complete or empty.
The exact values of the neighbor connectivity for the complete graphs $K_n$ and cycles $C_n$ are straightforward from the definition, namely $\kappa_{NB}(K_{n})=0$ for $n\geq 1$ while
\[
\kappa_{NB}(C_{n})=
\begin{cases}
0&\quad \text{ for $n=3$}\\
1&\quad \text{ for $4\leq n\leq 5$}\\
2&\quad \text{ for $n\geq 6$}\,.	
\end{cases}
\]

\subsection{$K$-ary $n$-cubes}

The \emph{Cayley graph}  of a finite group $(\Gamma, \cdot)$ with respect to its generating set $S$, denoted by $Cay(\Gamma,S)$, is defined as the graph with vertex set $\Gamma$  and edge set $\{\{g, g\cdot s\}\,|\,g\in \Gamma, s\in S\}$. Following \cite{D1}, we assume that for each generator $s\in S$, its inverse $s^{-1}$ also falls into $S$, and that the identity element $e$ is not included in $S$.  
If the group $\Gamma$ is abelian, $Cay(\Gamma, S)$ is referred to as an \emph{abelian Cayley graph}.
 
Given integers $k\geq2$ and $n\geq1$, the \emph{$k$-ary $n$-cube} $Q_{n}^{k}$ may be defined as the Cayley graph $Cay(\mathbb{Z}_k^n, \{e_i, e_i^{-1}\}_{i=1}^n)$ on the group $(\mathbb{Z}_k^n, \oplus)$ where $\mathbb{Z}_k^n$ stands for the $n$-th Cartesian power of $\mathbb{Z}_k=[k]$ while $\oplus$ represents coordinatewise addition $\bmod\ k$. The generators $e_i$ and $e_i^{-1}$ have one and $k-1$, respectively, in the $i$-th coordinate and zeroes elsewhere.
From a less abstract point of view, $Q_{n}^{k}$ is just the graph whose vertices are strings of the form 
$u_{n-1}u_{n-2}\cdots u_{0}$ where $u_{i}\in [k]$ for all $i\in [n]$. 
Two vertices $u_{n-1}u_{n-2}\cdots u_{0}$ and $v_{n-1}v_{n-2}\cdots v_{0}$ are adjacent whenever there is a~$j\in[n]$ such that $|u_{j}-v_{j}|\in\{1,k-1\}$ while $u_{i}=v_{i}$ for every $i\in [n]\setminus\{j\}$.

Note that $Q_n^k$ is a $\delta$-regular graph where $\delta=n$ for $k=2$ while $\delta=2n$ for $k\ge3$. Moreover, as $Q_n^k$ is a connected edge-transitive graph, we have $\kappa(Q_n^k)=\delta(Q_n^k)$ \cite[Problem 12.15]{Lovasz}. Consequently, $Q_{n}^{k}$ is $n$-connected for $k=2$ and $2n$-connected for $k\geq 3$ \cite{D97}.

Regarding the special cases $n=1$, $k=2$ and $n=2$: $Q_{1}^{k}$ is a cycle of length $k$, $Q_{n}^{2}$ is an {\it $n$-dimensional hypercube} and $Q_{2}^{k}$ is a $k\times k$ \emph{wrap-around mesh}. The graphs $Q_1^3$, $Q_2^3$ and $Q_3^3$ are shown in Fig.~\ref{fig1}.

\begin{figure}
\begin{center}
\includegraphics[width=5.5in]{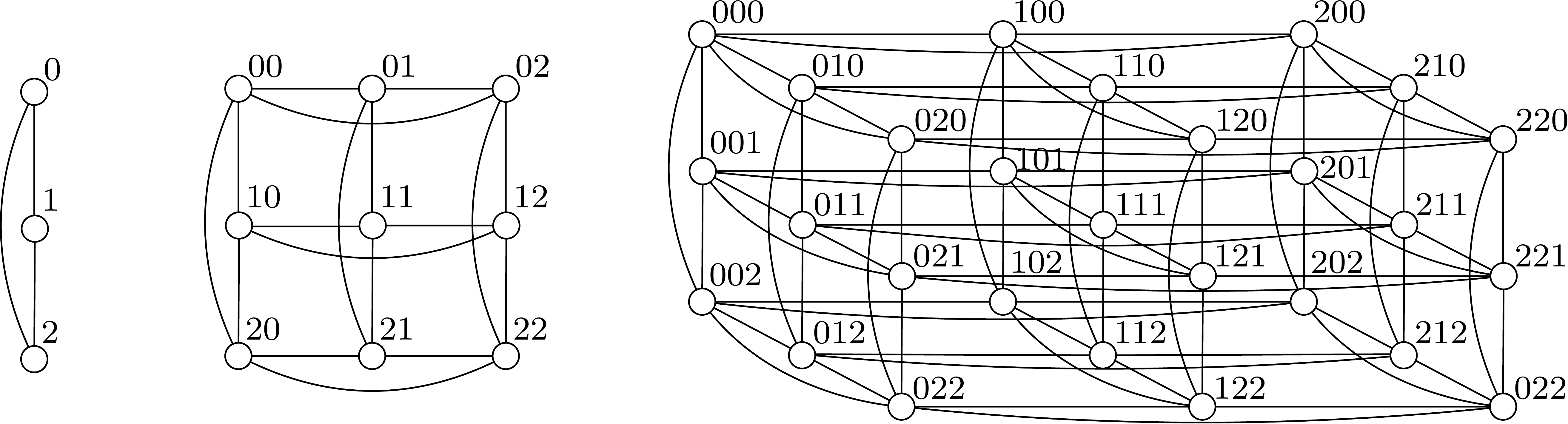} 
\end{center}
\caption {Illustration of $3$-ary $n$-cubes $Q_{1}^{3}$, $Q_{2}^{3}$ and $Q_{3}^{3}$}
\label{fig1}
\end{figure}

Now fix a $d\in[n]$ and note that for any $i\in[k]$, the subgraph of $Q_n^k$ induced by the set $\{u_{n-1}u_{n-2}\cdots u_{0}\in V(Q_n^k)\mid u_d=i\}$ is isomorphic to $Q_{n-1}^k$. We call this subgraph a \emph{subcube} of $Q_n^k$ and denote it by $Q[i]$. 
Given a vertex $u=u_{n-1}u_{n-2}\cdots u_{0}$ of $Q[i]$ and $j\in[k]$, we use $u^j$ to denote the vertex obtained from $u$ by changing the value of $u_d$ to $j$, i.e., the vertex $u_{n-1}\cdots u_{d+1}\, j\, u_{d-1} \dots  u_{0}$.
If $i,j\in[k]$ such that $|i-j|\in\{1,k-1\}$, $Q[i]$ and $Q[j]$ are called \emph{adjacent subcubes}. Each vertex $u$ of $Q[i]$ has a unique neighbor $u^j$ in $Q[j]$, called an \emph{outer neighbor} of $u$. 

Note that any two vertices of a hypercube have either two neighbors in common or none at all. This condition, known as the \emph{$(0,2)$-property}, can be even used to characterize hypercubes \cite[Section~2.2]{Mulder}. The following lemma \cite{G14,H12} explains that $k$-ary $n$-cubes inherit this property only for $k\in\{2,4\}$ (simply because $Q_n^2=Q_n$ and $Q_n^4\cong Q_{2n}$) while in the other cases there are also vertices with only one common neighbor. 
\begin{lemma}\label{lem:(0,2)} 
For any two distinct vertices $x,y\in V(Q_{n}^{k})$ we have
\[
 |N(x)\cap N(y)|\in\left\{\begin{array}{ll}
 \{0,2\}& \ \text{if $k\in\{2,4\}$}\\
 \{0,1,2\}&\ \text{otherwise}\, .
  \end{array}\right.
 \]
Moreover, if $k=3$, then  $|N(x)\cap N(y)|=1$ iff $x$ and $y$ are adjacent.
\end{lemma}

The last lemma of this section applies a counting argument to impose a lower bound on the number of healthy pairs of neighbors in adjacent subcubes. Although rather technical, this result provides a useful tool to be used later in our constructions of internally disjoint paths.
\begin{lemma}
\label{lem:counting}
Let $n,k\ge3$, $U$ be a set of $\ell\in[n]$ faults in $Q_n^k$, $Q[i]$ and $Q[j]$ be adjacent subcubes of $Q_n^k$ and $x$ be a healthy vertex of $Q[i]$. Put $u_i=|U\cap V(Q[i])|$, $u_j=|U\cap V(Q[j])|$ and $h = 2n-2-\ell-u_i-u_j$. Then we have
\[
    |\{v \in V(Q[i]) \mid \text{ both $v$ and $v^j$ are healthy}\}| > h
\]
and, moreover, at least $h$ of these vertices are neighbors of $x$. 
\end{lemma}
\begin{proof}
Put $U_i=U\cap V(Q[i])$, $U_j=U\cap V(Q[j])$ and $H=\{v \in V(Q[i]) \mid v,v^j\not\in N[U]\}$. First we show that $H$ contains at least $h$ neighbors of $x$. To that end, observe that the vertex $x$ has $2n-2$ neighbors in $Q[i]$ and, by Lemma~\ref{lem:(0,2)}, at most $2u_i$ of them may fall into $N[U_i]$. Similarly, at most $2u_j$ neighbors of $x^j$ in $Q[j]$ may fall into $N[U_j]$. Moreover, if a vertex of $Q[i]$ falls into $N[U_{j}]$, then its outer neighbor in $Q[j]$ falls into $U_{j}$, and an analogous statement holds for the vertices of $Q[j]$. Finally, the number of vertices $v$ of  $Q[i]$ such that $v$ or $v^j$ falls into $N(U\setminus (U_i\cup U_j))$ does not exceed $\ell-u_i-u_j$.  It follows that there are at least  
\[
2n-2-2u_i-2u_j-(\ell-u_i-u_j)=h
\] 
healthy neighbors of $x$ in $Q[i]$ such that their outer neighbors  in $Q[j]$ are healthy as well, i.e., $H$ contains at least $h$ neighbors of $x$.

It remains to prove that $|H|>h$. Put $F=V(Q[i])\setminus H$ and note that as $|F|+|H|= |V(Q[i])|$, it suffices to show that $|F|+h<|V(Q[i])|$. Recall that $Q[i], Q[j]$ are regular graphs of degree $2n-2$ and therefore $|F| \le \ell + (u_i+u_j)(2n-2)$.  It follows that
\[
|F| +h\le (u_i+u_j+1)(2n-3) + 1\le 
\begin{cases} n(2n-3)+1 &<  k^{n-1} = |V(Q[i])|  \ \text{for } k\ge3, n\ge4 \\
\hfill 7 &< 9\phantom{{}^{n-1}} = |V(Q[i])| \  \text{for } k=n=3, u_i+u_j\le1 ,
\end{cases}
\]
using $u_i+u_j\le\ell\le n-1$ to obtain the second inequality. 
Regarding the case that $k=n=3$ and $u_i+u_j=2$, note that then $Q[i]$ is isomorphic to $Q_2^3$, see Fig.~\ref{fig1}. But then Lemma~\ref{lem:(0,2)} implies that $|N[x]\cap N[y]|\ge2$ for any two distinct vertices $x,y\in F$, which means that actually  $|F|\le \ell + (u_i+u_j)(2n-2)-2$.
Hence  in this case we have 
\[
|F|+h \le (u_i+u_j+1)(2n-3) -1 = 8 < 9 = |V(Q[i])| .
\]
It follows that $|H|>h$ for all $n,k\ge3$ and the proof is complete.
\end{proof}

\section{Hypercubes}
\label{sec:hypercubes}
In order to determine $\kappa_{NB}(Q^k_{n})$, we need to find the smallest set $U\subseteq V(Q_n)$ such that $Q^k_n\ominus U$ becomes disconnected. To that end, we derive a more general statement, providing lower bounds to the connectivity of $Q^k_n\ominus U$ for all subsets $U$ whose cardinality does not exceed $\lceil\delta(Q_n^k)/2\rceil$. 
We first deal with the class of hypercubes.   
\begin{theorem}\label{th:hypercube}
Let $n\ge2$ and $0\leq \ell\leq \lceil n/2\rceil$. Then  for every set $U\subseteq V(Q_n)$ of $\ell$ faults we have 
\[
\kappa (Q_n\ominus U)\geq \max\{n-2\ell,0\} .
\]
\end{theorem}
\begin{proof}
First note that as 
$Q_2$ is a 4-cycle, $\kappa(C_4)=2$ and $\kappa(K_1)=0$, the theorem holds for $n=2$. Hence we can assume that $n\ge3$.

We argue by induction on $n$.  Note that it suffices to settle the case that $0<\ell<n/2$, for otherwise the statement of the theorem is either trivial (if $\ell\ge n/2$) or follows from $\kappa(Q_n)=n$ (if $\ell=0$). Assume that $n\ge3$ and that the statement holds for every set $U\subseteq V(Q_{n-1})$, $0\le|U|\le \lceil (n-1)/2\rceil$.

Let $U\subseteq V(Q_n)$ be a set of $\ell$ faults, $0<\ell<n/2$. To prove that $\kappa (Q_n\ominus U)\geq n-2\ell$, by Menger's Theorem it suffices to show that for any distinct $x,y\in V(Q_n\ominus U)$ there are at least $n-2\ell$ internally disjoint $(x,y)$-paths in $Q_n\ominus U$.
Note that $Q_{n}$ can be partitioned into two subcubes $Q[0]$ and $Q[1]$ such that $x$ and $y$ are in distinct subcubes. Put $U_i=U\cap V(Q[i])$ and $u_i=|U_i|$ for both $i\in[2]$. Without loss of generality assume that $x\in V(Q[0])$, $y\in V(Q[1])$ and $u_0\ge u_1$.

Recall that as $x\in V(Q_n\ominus U)$, no neighbor of $x$ belongs to $U$, but some of them may fall into $N(U)$.
Observe that $x$ has $n-1$ neighbors in $Q[0]$ and, by the $(0,2)$-property (Lemma~\ref{lem:(0,2)}), at most $2u_0$ of them may belong to $N(U_0)$.  Hence there are neighbors $x_1,x_2,\ldots,x_{n-1-2u_0}$ of $x$ in $Q[0]\ominus U_0$. Consider the outer neighbors $x^1,x_1^1,x_2^1,\ldots,x_{n-1-2u_0}^1$ of $x, x_1, x_2,\ldots,x_{n-1-2u_0}$ in $Q[1]$ and distinguish two cases.

First settle the case that $x^1\not\in N[U_1]$. Using the $(0,2)$-property, observe that at most $2u_1$ neighbors of $x^1$ in $Q[1]$ may fall into $N[U_1]$. We can therefore without loss of generality assume that $x_1^1,x_2^1,\ldots,x_{n-1-2\ell}^1$ belong to $Q[1]\ominus U_1-N(U_0)$. Note that then their outer neighbors $x_1, x_2, \dots, x_{n-1-2\ell}$ belong to $Q[0]\ominus U_0-N(U_1)$, i.e., all these vertices are healthy.

Put $Y=\{x^1,x_1^1,x_2^1,\ldots,x_{n-1-2\ell}^1\}$ and note that the only faulty vertices in $Q[1]\ominus U_1$ are the outer neighbors of $U_0$ and there are at most $u_0$ of them. As $Q[1]\ominus U_1$ is $(n-1-2u_1)$-connected by the induction hypothesis and   
\[
|Y|+u_0=n-u_0-2u_1\le n-1-2u_1 ,
\]
using $u_0\ge \lceil \ell/2 \rceil\ge1$ to obtain the last inequality,
by Lemma~\ref{lem:fan} there exists a $(Y,y)$-fan in $Q[1]\ominus U_1-N(U_0)$, i.e.,  a family of $n-2\ell$ internally vertex-disjoint $(Y,y)$-paths $P_{x^1y}$, $P_{x_1^1y}$, $P_{x_2^1y}$, $\dots, P_{x_{n-1-2\ell}^1y}$. Then $P_0=(x,P_{x^1y})$ and
$P_i=(x,x_i,P_{x_i^1y})$ for $1\leq i\leq n-1-2\ell$ form the required family of
$n-2\ell$ internally disjoint $(x,y)$-paths in $Q_n\ominus U$.

It remains to settle the case that $x^1 \in N[U_1]$ which means that $u_1\ge1$ and some neighbor of $x^1$ in $Q[1]$ belongs to $U_1$. 
Then at most $2u_1-1$ neighbors of $x^1$ in $Q[1]$ may fall into $N[U_1]$, which means that we can without loss of generality assume that $x_1^1,x_2^1,\ldots, x_{n-2\ell}^1$ belong to $Q[1]\ominus U_1-N(U_0)$. By the same argument as in the previous case, there is a family of $n-2\ell$ internally disjoint $(x_i^1,y)$-paths $P_{x_i^1y}$, $1\le i \le n-2\ell$, in $Q[1]$ avoiding $N[U]$. Then $P_i=(x, x_i, P_{x_i^1y})$ for $1\leq i\leq n-2\ell$ form the required family of $n-2\ell$ internally disjoint $(x,y)$-paths in $Q_n\ominus U$.
\end{proof}
Note that the proof above may be viewed as a sketch of the proof for the case of $Q_n^k$ with $k\ge3$, which is based on similar ideas, only operates in a more general setting and the constructions are therefore more involved. The details are provided in the next section.
\section{$K$-ary $n$-cubes with $k\ge3$}
\label{sec:general}
We start with a lemma that constructs paths between endvertices in adjacent subcubes. It will be useful in the general case to compose paths between endvertices at an arbitrary position.
\begin{lemma}
\label{lem:adjacent}
Let $n,k\ge3$, $U\subseteq V(Q_n^k)$ be a set of $\ell$ faults, $0<\ell<n$, and $Q[j], Q[j']$ be adjacent subcubes of $Q_n^k$ such that $\kappa(Q[i]\ominus U_i)\ge 2n-2-2u_i$ where $U_i=U\cap V(Q[i])$ and $u_i=|U_i|$ for both $i\in\{j,j'\}$. Then for  arbitrary healthy vertices $x$ and $y$ of $Q[j]$ and $Q[j']$, respectively, there are at least $m$ internally disjoint $(x,y)$-paths passing only through healthy vertices of $Q[j]$ and $Q[j']$, where
\[
m = 
\begin{cases}
 2n-2\ell \quad &\text{if\ \,$u_j+u_{j'}<\ell $} \\
 2n-2\ell-1\quad &\text{if\  \,$u_j+u_{j'}=\ell $}\, .
\end{cases}
\] 
\end{lemma}
\begin{proof}
Suppose that $Q_{n}^k$ is partitioned into subcubes $Q[0],Q[1],\ldots, Q[k-1]$. We can     without loss of generality assume that $j=0$, $j'=1$ and $u_0\ge u_1$.
Observe that by Lemma~\ref{lem:counting} there are $h=2n-2-u_0-u_1-\ell$ healthy neighbors $x_1, x_2, \dots, x_{h}$ of $x$ in $Q[0]$ such that their outer neighbors $x_1^1, x_2^1, \dots, x_{h}^1$ in $Q[1]$ are healthy as well. Put 
$U_i=U\cap V(Q[i])$ and $u_i=|U_i|$ for all $i\in[k]$, $X=\{x_i\}_{i=1}^h$, $Y=\{x_i^1\}_{i=1}^h$ and consider three cases.  

\textsc{Case 1}. $u_0=\ell=n-1$.

Note that in this case we have $u_i=0$ for all $i\in[k]\setminus\{0\}$ while $m=1$. As $x\not\in U_0=U$, its outer neighbor $x^1$ in $Q[1]$ must be healthy.  Moreover, the only faulty vertices in $Q[1]$ are the $n-1$ outer neighbors of $U_0$. As $Q[1]$ is $(2n-2)$-connected and $2n-2-(n-1)>0$, by Lemma~\ref{lem:fan} there is an $(x^1,y)$-path $P_{x^1y}$ in $Q[1]$ avoiding $N[U]$. Then $P=(x, P_{x^1y})$ is the desired $(x,y)$-path  passing only through healthy vertices of $Q[0]$ and $Q[1]$.

\textsc{Case 2}. $1\le u_0\le n-2$.

By Lemma~\ref{lem:counting}, $Q[0]$ contains a healthy vertex $z\not\in X$ such that its outer neighbor $z^1$ in $Q[1]$ is healthy as well.  We claim that then $Q[0]$ contains an $(x,z)$-path $P_{xz}$ avoiding both faulty vertices and all vertices of $X$.

To prove the claim, note that if $x^1$ is healthy, it suffices to set $z:=x$, which means that $P_{xz}=P_{xx}=(x)$. On the other hand, the case that $x^1$ is faulty requires more attention. Then we have $z\ne x$ and $u_1+u_2\ge1$. Note that the only faulty vertices that may fall into $Q[0]\ominus U_0$ are at most $u_1+u_{k-1}$ outer neighbors of $U_1\cup U_{k-1}$ and 
\[
|X|+1+u_1+u_{k-1} = 2n-1-2u_0-\sum_{i=1}^{k-2}u_i\le 2n-2-2u_0 
\] 
(using $u_1+u_2\ge1$ to obtain the last inequality). As we assume that $Q[0]\ominus U_0$ is $(2n-2-2u_0)$-connected and $2n-2-2u_0>0$,  by Lemma~\ref{lem:fan} there is an $(x,z)$-path $P_{xz}$ in $Q[0]\ominus U_0$ avoiding both $N(U_1\cup U_{k-1})$ and $X$. This completes the verification of the claim.

Further, put $Y'=Y\cup\{z^1\}$ and note that all vertices of $Y'$ are healthy.  Moreover, the only faulty vertices that may fall into $Q[1]\ominus U_1$ are outer neighbors of $U_0$ and $U_2$ and therefore their number does not exceed $u_0+u_2$. Since $Q[1]\ominus U_1$ is $(2n-2-2u_1)$-connected by our assumption,
\[
|Y'|+u_0+u_2 = 2n - 1 - u_0 -2u_1 -\sum_{i=3}^{k-1}u_i\le 2n-2-2u_1
\] 
(using $u_0\ge1$  for the last inequality) and and $2n-2-2u_1>0$, by Lemma~\ref{lem:fan} there is a $(Y',y)$-fan in $Q[1]\ominus U_1-N(U_0\cup U_2)$, i.e.,  a family of $h+1$ internally disjoint paths $P_{x_1^1y},P_{x_2^1y}, \dots, P_{x_h^1y}, P_{z^1y} $, passing only through healthy vertices of $Q[1]$.  
It remains to set $P_i=(x,x_i,P_{x_i^1y})$ for $1\leq i\leq h$, $P_{h+1}=(P_{xz}, P_{z^1y})$ and observe that  then $\{P_i\}_{i=1}^{h+1}$ form a family of $h+1$ internally disjoint $(x,y)$-paths passing only through healthy vertices of $Q[0]$ and $Q[1]$. 
As 
\[
h+1=2n-u_0-u_1-\ell-1
\begin{cases}
\ge 2n-2\ell \quad &\text{if $u_0+u_1<\ell$} \\
= 2n-2\ell-1  \quad &\text{if $u_0+u_1=\ell$} \, ,	
\end{cases}
\] 
it follows that $h+1\ge m$ and we are done in this case.

\textsc{Case 3}. $u_0=0$.

Note that in this case we have $u_0=u_1=0<\ell$ and therefore $m=2n-2\ell$. Moreover, the only faulty vertices in $Q[1]$ are at most $l$ outer neighbors of $U_2$. Since $Q[1]$ is $(2n-2)$-connected and
\[
|Y|+\ell = 2n-2 ,
\] 
by Lemma~\ref{lem:fan} there is a $(Y,y)$-fan in $Q[1]-N(U_2)$  consisting of $h$ paths $P_{x_1^1y}$, $P_{x_2^1y}$, $ \dots, P_{x_{h}^1y}$. Set $P_i=(x,x_i,P_{x_i^1y})$ for $1\leq i\leq h$ and observe that then $\{P_i\}_{i=1}^{h}$ form a~family of 
\[
h=2n-2-\ell
\] 
internally disjoint $(x,y)$-paths passing only through healthy vertices of $Q[0]$ and $Q[1]$. If $\ell\ge2$, then $h\ge 2n-2\ell$ and we are done.

It remains to settle the case $\ell=1$. First assume that $k\ge4$. Then either $U_2$ or $U_{k-1}$ must be empty and we can without loss of generality assume the former, i.e. $u_2=0$ while $u_{k-1}\le1$. Note that then all vertices of $Q[1]$ must be healthy. 
Further, put $Y'=Y\cup\{x^1\}$ and note that as $Q[1]$ is $(2n-2)$-connected and
\[
|Y'| = h+1 =2n-2 ,
\] 
by Lemma~\ref{lem:fan} there is a $(Y',y)$-fan in $Q[1]$  consisting of $h+1$  paths $P_{x_1^1y},P_{x_2^1y}, \dots, P_{x_{h}^1y}, P_{x^1y} $. Put $P_i=(x,x_i,P_{x_i^1y})$ for $1\leq i\leq h$, $P_{h+1}=(x, P_{x^1y})$ 
and observe that then $\{P_i\}_{i=1}^{h+1}$ form the desired family of $h+1=2n-2\ell$ internally disjoint $(x,y)$-paths passing only through healthy vertices of $Q[0]$ and $Q[1]$.

If $\ell=1$ and $k=3$, then $u_2=1$ while both $x^1$ and $y^0$ are healthy.  
Moreover, if $y^0\in X$, we can without loss of generality assume that $y^0=x_h$.
Put $X'=X\setminus\{x_h\}$ and observe that as $|X|=h=2n-3 > 0$, we have $|X'|=h-1$.
Recall that the only faulty vertex in $Q[0]$ is the outer neighbor of $U_2$,  
\[
|X'|+2 =2n-2\,,
\] 
$Q[0]$ is $(2n-2)$-connected and $y^0\not\in N[U]\cup X'$. Consequently,  
by Lemma~\ref{lem:fan} there is an $(x,y^0)$-path $P_{xy^0}$ in $Q[0]$ avoiding both $N[U]$ and $X'$.

Further, put $Y'=(Y\setminus\{x^1_h\})\cup\{x^1\}$ and recall that the only faulty vertex in $Q[1]$ is the outer neighbor of $U_2$. Since $Q[1]$ is $(2n-2)$-connected and
\[
|Y'|+1 = 2n-2 ,
\] 
by Lemma~\ref{lem:fan} there is a $(Y',y)$-fan in $Q[1]-N(U_2)$  consisting of $h$  paths $P_{x_1^1y},P_{x_2^1y}, \dots, P_{x_{h-1}^1y}, P_{x^1y} $. Set $P_i=(x,x_i,P_{x_i^1y})$ for $1\leq i\leq h-1$, $P_{h}=(x, P_{x^1y})$, $P_{h+1}=(P_{xy^0}, y)$, and observe that  then $\{P_i\}_{i=1}^{h+1}$ form the desired family of 
\[
h+1=2n-2=2n-2\ell
\] 
internally disjoint $(x,y)$-paths passing only through healthy vertices of $Q[0]$ and $Q[1]$.
\end{proof}
Now we are ready to state an analogy of Theorem~\ref{th:hypercube} for $k$-ary $n$-cubes with $k\ge3$. 
\begin{theorem}\label{th2}
Let $n\ge2$, $k\geq $3 and $0\leq \ell\leq n$. Then for every set $U\subseteq V(Q_n^k)$ of $\ell$ faults we have 
\[
\kappa (Q_n^k\ominus U)\geq 2n-2\ell .
\]
\end{theorem}
\begin{proof}
First note that it suffices to verify only the case that $0<\ell<n$, for otherwise the statement  of the theorem is either trivial (if $\ell = n$) or follows from $\kappa(Q_n^k)=2n$ (if $\ell=0$). 

We argue by induction on $n$.  To settle the case  $n=2$, recall that $Q_{2}^{k}$ is a $k\times k$ wrap-around mesh. If we delete an arbitrary vertex $v$ and its neighbors, the resulting  graph $Q_2^k\ominus v$ contains a~spanning subgraph, consisting of a $(k-1)\times (k-1)$ mesh $M$ and (if $k>3$) two paths of length $k-2$ whose endvertices are identified with vertices of $M$. Such a spanning subgraph is clearly 2-connected, which means that $\kappa(Q_2^k\ominus v)\ge2$ as claimed by the theorem. 
Assume that $n\ge3$ and that the statement of the theorem  holds for every set $U\subseteq V(Q_{n-1}^k)$, $0\le|U|\le n-1$.

Let $U\subseteq V(Q_n^k)$ be a set faults of size $\ell$. To prove that $\kappa (Q_n^k\ominus U)\geq 2n-2\ell$, we only need to show that for any distinct $x,y\in V(Q_n^k\ominus U)$ there are at least $2n-2\ell$ internally disjoint $(x,y)$-paths in $Q_n^k\ominus U$. Note that $Q_{n}^k$ can be partitioned into $k$ subcubes $Q[0],Q[1],\ldots, Q[k-1]$ such that $x$ and $y$ are in distinct subcubes. Put $U_i=U\cap V(Q[i])$  and $u_i=|U_i|$ for all $i\in[k]$.
Recall that by the induction hypothesis, $\kappa(Q[i]\ominus U_i)\ge2n-2-2u_i$ for all $i\in[k]$, which means that  the assumptions of Lemma~\ref{lem:adjacent} are satisfied and the lemma is applicable to any pair of adjacent subcubes.
Without loss of generality assume that $x\in V(Q[0])$ and distinguish two cases depending on whether $x$ and $y$ are in adjacent subcubes or not.

\textsc{Case 1}. Vertices $x$ and $y$ are in adjacent subcubes. 
We can without loss of generality assume that $y\in V(Q[1])$ and $u_0\ge u_1$.

By Lemma~\ref{lem:adjacent} there are $m$ internally disjoint $(x,y)$-paths $P_1,P_2,\dots,P_m$, passing only through healthy vertices of $Q[0]$ and $Q[1]$. If $u_0+u_1<\ell$, then $m=2n-2\ell$ and we are done. If $u_0+u_1=\ell$, which means that $m=2n-2\ell-1$, we need to construct an additional $(x,y)$-path in $Q_n^k\ominus U$, internally disjoint with $P_1, P_2, \dots, P_m$. 

Note that our assumption that $u_0+u_1=\ell$ implies $u_i=0$ for all $1<i\le k-1$.
As $y$ is a~healthy vertex of $Q[1]$, it follows that $y^0\not\in U_0$ and therefore the outer neighbor $y^2$ of $y$ in $Q[2]$ is also healthy. Similarly, as $x$ is a healthy vertex of $Q[0]$, all of $x^1,x^2,\dots, x^{k-1}$ must be healthy as well.
Further, as $Q[2]$ is $(2n-2)$-connected, it contains at most $\ell$ faulty vertices and $\ell \le n-1<2n-2$, by Lemma~\ref{lem:fan} there is an $(x^2,y^2)$-path $P_{x^2y^2}$ in $Q[2]$ avoiding $N[U]$. It remains to construct
an additional path 
\[
P_{m+1}=
\begin{cases}
(x,P_{x^2y^2},y)\quad & \text{if $k=3$}\\
(x, x^{k-1}, x^{k-2}, \dots, x^3, P_{x^2y^2},y) \quad & \text{if $k\ge4$} 
\end{cases}
\]
and observe that  then $\{P_i\}_{i=1}^{m+1}$ form the desired family of $m+1=2n-2\ell$ internally disjoint $(x,y)$-paths in $Q_n^k\ominus U$.

\textsc{Case 2}. Vertices $x$ and $y$ belong to non-adjacent subcubes. 

Note that then $k\ge4$. Without loss of generality assume that $y\in V(Q[t])$, $2\le t \le k-2$, and $u_0\ge u_t$. Further, note that Lemma~\ref{lem:counting} implies that each of the subcubes $Q[i]$
contains a healthy vertex $v_i$ for $1\le i \le t-1$. Put $v_0:= x$, $v_{t}:=y$, apply Lemma~\ref{lem:adjacent} to each pair of vertices $v_i, v_{i+1}$ in adjacent subcubes $Q[i], Q[i+1]$, $i\in[t]$ and from the resulting $(v_i,v_{i+1})$-paths $\{P^j_{v_iv_{i+1}}\}_{j=0}^m$ extract   
\begin{enumerate}[\upshape(1)]
	\item prefixes of maximum length fully contained in $Q[0]$ (for $i=0$), i.e., prefixes $\{P_{v_0v_{0j}}\}_{j=0}^m$ of $\{P^j_{v_0v_{1}}\}_{j=0}^m$, respectively, such that $v_{0j}$ is the first vertex of $P^j_{v_0v_{1}}$ whose successor in $P^j_{v_0v_{1}}$ belongs to $Q[1]$;
\item suffixes of maximum length fully contained in $Q[t]$ (for $i=t-1$), i.e., suffixes $\{P_{u_{tj}v_{t}}\}_{j=0}^m$ of $\{P^j_{v_{t-1}v_{t}}\}_{j=0}^m$, respectively, such that $u_{tj}$ is the last vertex of $P^j_{v_{t-1}v_{t}}$ whose predecessor in $P^j_{v_{t-1}v_{t}}$ belongs to $Q[t-1]$;
\item edges between $Q[i]$ and $Q[i+1]$ (for $1\le i \le t-2$), i.e., edges $\{v_{ij} u_{i+1j}\}_{j=0}^{m-1}$ formed by consecutive vertices of $\{P^j_{v_iv_{i+1}}\}_{j=0}^m$, respectively, such that $v_{ij}\in V(Q[i])$ while $u_{i+1j}\in V(Q[i+1])$.
\end{enumerate} 
Consequently, it follows that there are
\begin{enumerate}[\upshape(1)]
\item \label{th2:1}
internally disjoint $(x, v_{0j})$-paths $P_{xv_{0j}}$ for $j\in[m]$ passing only through healthy vertices of $Q[0]$ such that the outer neighbors $\{u_{1j}\}_{j=0}^{m-1}$ of $\{v_{0j}\}_{j=0}^{m-1}$ in $Q[1]$ are healthy as well;
\item \label{th2:2} 
internally disjoint $(u_{tj},y)$-paths $P_{u_{tj}y}$ for $j\in[m]$ passing only through healthy vertices of $Q[t]$ such that such that the outer neighbors $\{v_{t-1j}\}_{j=0}^{m-1}$ of $\{u_{tj}\}_{j=0}^{m-1}$ in $Q[t-1]$ are healthy as well;
\item \label{th2:3} 
healthy vertices $\{v_{ij}\}_{j=0}^{m-1}$ in $Q[i]$ for each $1\le i\le t-2$ such that their outer neighbors $\{u_{i+1j}\}_{j=0}^{m-1}$ in $Q[i+1]$ are healthy as well,
\end{enumerate}
where $m$ is defined in the following way: If there is an $i\in[t]$ such that $u_i+u_{i+1}=\ell$, then $m:=2n-2\ell-1$, otherwise $m:=2n-2\ell$. 

\textsc{Subcase 2.1}. $u_0<\ell$.

Recall that this assumption means that $u_t\le u_0<\ell$ as well. Hence we can assume that $u_{i-1}+u_i+u_{i+1}<\ell$ for all $0<i<t$, for otherwise the sequence $Q[0], Q[1], \dots, Q[t]$ of adjacent subcubes may be replaced with $Q[0], Q[k-1], Q[k-2], \dots, Q[t]$. Note that then $m=2n-2\ell$.

Now fix an $0<i<t$ and put $X_i=\{u_{ij}\}_{j=0}^{m-1}$, $Y_i=\{v_{ij}\}_{j=0}^{m-1}$. Observe that $Q[i]\ominus U_i$ is $(2n-2-2u_i)$-connected by the induction hypothesis, it contains at most $u_{i-1}+u_{i+1}<\ell-u_i$ faulty vertices, and
\[
2n-2-2u_i-(\ell-u_i-1)\ge 2n-2\ell ,
\]
using our assumption that $u_i<\ell$.
Hence by Lemma~\ref{lem:fan} there is a family of $m=2n-2\ell$ pairwise disjoint $(X_i,Y_i)$-paths passing only through healthy vertices of $Q[i]\ominus U_i$. Concatenating these paths for all $0<i<t$ with $\{P_{xv_{0j}}\}_{j=0}^{m-1}$ and $\{P_{v_{tj}y}\}_{j=0}^{m-1}$ over the respective endvertices, we obtain the desired family of $m=2n-2\ell$ internally disjoint $(x,y)$-paths in $Q_n^k\ominus U$. 

\textsc{Subcase 2.2}. $u_0=\ell$. 

Similarly as in the previous subcase, use \eqref{th2:1}-\eqref{th2:3} and Lemma~\ref{lem:fan} to set up a family of $m=2n-2\ell-1$ internally disjoint paths. However, this time construct only $(x,v_{t-1j})$-paths, $j\in[m]$, passing only through healthy vertices of $Q[0], Q[1], \dots, Q[t-1]$, where $\{v_{t-1j}\}_{j=0}^{m-1}$ are healthy vertices of $Q[t-1]$ such that their outer neighbors $\{u_{tj}\}_{j=0}^{m-1}$ in $Q[t]$ are healthy as well. Next, note that by Lemma~\ref{lem:counting} there are at least $2n-1-\ell-u_t-u_{t+1}=2n-\ell-1$ healthy vertices in $Q[t]$ such that their outer neighbors in $Q[t+1]$ are also healthy. Since
\[
2n-\ell-1>2n-2\ell-1=m \,,
\]
it follows that there is a $z\in V(Q[t])\setminus\{u_{tj}\}_{j=0}^{m-1}$ such that both $z$ and $z^{t+1}$ are healthy.

Next, note that the only faulty vertices in $Q[t+1]$ may be $\ell$ outer neighbors of $U_0$. As $Q[t+1]$ is $(2n-2)$-connected and $2n-2-\ell\ge n-1>0$, by Lemma~\ref{lem:fan} there is a $(x^{t+1},z^{t+1})$-path $P_{x^{t+1}z^{t+1}}$ passing only through healthy vertices of $Q[t+1]$.   

Finally, as $Q[t]$ contains no faulty vertices, it is $(2n-2)$-connected and $m+1=2n-2\ell\le2n-2$, by Lemma~\ref{lem:fan} there is a $(\{u_{tj}\}_{j=0}^{m-1} \cup \{z\}, y)$-fan in $Q[t]$, consisting of $(u_{tj},y)$-paths
$P_{u_{tj}y}$ for $j\in[m]$ and a $(z,y)$-path $P_{zy}$. It remains to set
$P_j=(P_{xv_{t-1j}},P_{u_{tj}y})$ for $j\in[m]$, $P_m=(x,x^{k-1},x^{k-2},\dots,x^{t+2},P_{x^{t+1}z^{t+1}},P_{zy})$
and conclude  that $P_0,P_1, \dots, P_m$ form the desired family of $m+1=2n-2\ell$ internally disjoint $(x,y)$-paths in $Q_n^k\ominus U$. The proof is complete.
\end{proof}
\section{Main results}
\label{sec:main}
\begin{theorem}
\label{thm:cayley}
Let $G=Cay(\Gamma,S)$ be an abelian Cayley graph of degree $\delta$. If there is an ordering $s_1,s_2, \dots,s_{\delta}$ of all generators in $S$ such that $s_{2i-1} \cdot s_{2i} \not\in S\cup\{e\}$ for all $1\le i \le \lfloor \delta/2 \rfloor$, then 
\[
\kappa_{NB}(G)\le\lceil \delta/2 \rceil \,.
\]	
\end{theorem}
\begin{proof}
Select the vertex $e$ of $G$ (where $e$ is the identity element of $\Gamma$) and note that then $N(e)=\{s_1,s_2, \dots,s_{\delta}\}$. 
Let $T$ denote the set of all vertices of $G$ at distance two from $e$.
Note that for every $i\in\{1,2,\dots,\lfloor \delta/2\rfloor\}$ there is a vertex $v_i=s_{2i-1}\cdot s_{2i}$, which, by our assumption, does not belong to $N[e]$, and therefore it must fall into $T$. Moreover, as the group $\Gamma$ is abelian, we have $v_i=s_{2i-1}\cdot s_{2i}= s_{2i} \cdot s_{2i-1}$, which means that $v_i$ is adjacent to both $s_{2i-1}$ and $s_{2i}$. Let $U'$ be the set of all these (not necessarily distinct) vertices $v_1, v_2, \dots, v_{\lfloor \delta/2\rfloor}$. 
If $\delta$ is odd and $N(s_\delta)\cap T\ne\emptyset$, let $v$ be an arbitrary neighbor of $s_\delta$ in $T$. 
Put
\[
U=\begin{cases}
        U' \cup \{v\} &\text{if $\delta$ is odd and $N(s_\delta)\cap T\ne\emptyset$} \\
        U' &\text{otherwise}.
      \end{cases}
\]
Note that then we have $N[e]\setminus N[U]=\{e\}$ in the former case while $N[\{e,s_\delta\}]\setminus N[U]\in\{\{e\},\{e, s_\delta\}\}$ in the latter.   
It follows that $G\ominus U$ contains either an isolated vertex $e$, or a component consisting of $e$ and $s_\delta$ joined by an edge. Hence $G\ominus U$ is either disconnected or complete, and therefore  $\kappa_{NB}(G)\le|U|\le\lceil \delta/2 \rceil$.
\end{proof}

\begin{theorem}
Let $n\ge1$ and $k\ge2$. Then
\[
\kappa_{NB}(Q_n^k)=
\begin{cases} 
0 & \text{for $n=1$ and $2\le k\le3$}\\
2 & \text{for $n=1$ and $k\geq 6$}\\
\lceil n/2\rceil &\text{for $n\ge2$ and $k = 2$}\\ 
n &\text{otherwise} .
\end{cases}
\]
\end{theorem}
\begin{proof}
Recall that 
\[Q_1^k=
\begin{cases}
K_2&\quad \text{for $k=2$}\\
C_k&\quad \text{for $k\ge3$}	
\end{cases}
\]
and therefore the values of $\kappa_{NB}(Q_1^k)$ follow directly from the definition. 
Hence we can assume that $n\ge2$. 
Put $\delta=\delta(Q_n^k)$, i.e., $\delta=n$ for $k=2$ while $\delta=2n$ for $k\ge3$. Recall that $Q_n^k$ is an abelian Cayley graph $Cay(\mathbb{Z}_k^n,S)$ and we claim that it satisfies the assumption of Theorem~\ref{thm:cayley}. Indeed,  if $k=2$, then $S=\{e_i\}_{i=1}^n$ (because $e_i=e_i^{-1}$ for all $1\le i\le n$ in this case), and an~arbitrary ordering of the generators in $S$ satisfies the assumption of Theorem~\ref{thm:cayley}. If $k\ge3$, then $S=\{e_i, e_i^{-1}\}_{i=1}^n$ and any ordering such that $e_i$ and $e_i^{-1}$ are not next to each other satisfies the assumption  of Theorem~\ref{thm:cayley}. Hence by this theorem, $\kappa_{NB}(Q_n^k)\le\lceil \delta/2 \rceil$ for $n\ge2$. 

On the other hand, if $U$ is a subset of $V(Q_n^k)$ is size $\ell<\lceil \delta/2 \rceil$, then by Theorems~\ref{th:hypercube} and \ref{th2} $Q_n^k\ominus U$ remains connected. Moreover, as $Q_n^k$ is a $\delta$-regular graph, for $\ell<\lceil \delta/2\rceil $ we have 
\[
|V(Q_n^k\ominus U)|\ge k^n-(\delta+1)\ell\ge k^n-(\delta+1)(\lceil \delta/2\rceil-1)>3
\]
 provided that  $n\ge2$ and $k\geq 2$. As $k$-ary $n$-cubes are $K_4$-free graphs, it follows that  $Q_n^k\ominus U$ can never be a complete graph or empty in this case. Hence $\kappa_{NB}(Q_n^k)\ge\lceil \delta/2 \rceil$ for $n\ge2$ and the statement of the theorem follows.
\end{proof}

\begin{corollary}
\label{cor}
Let $n\ge1$, $k\ge2$ and $\delta=\delta(Q_n^k)$. Then 
$\kappa_{NB}(Q_n^k)=\lceil \delta/2\rceil$
unless 
\begin{itemize}
\item $Q_n^k$ is a cycle of length at least $6$, then $n=1$, $k\geq 6$ and $\kappa_{NB}(Q_n^k)=\delta/2+1$, or
\item $Q_n^k$ is a complete graph, then $n=1$, $2\le k\le3$ and $\kappa_{NB}(Q_n^k)=0$.
\end{itemize}
\end{corollary}
\section{Concluding remarks}
Doty in \cite{D1} showed that the neighbor connectivity of abelian Cayley graphs of degree $\delta$ is bounded above by $\lceil\delta/2\rceil+2$, conjectured that the actual upper bound is  $\lceil\delta/2\rceil$ for all $\delta$-regular abelian Cayley graphs except cycles, and asked about Cayley graphs that achieve the maximum predicted neighbor connectivity. In this paper we partially answered the latter question by showing that a nontrivial class of abelian Cayley graphs, formed by the $k$-ary $n$-cubes, reaches the  maximum predicted value except for some trivial cases (Corollary~\ref{cor}).

It should be noted that there are abelian Cayley graphs whose neighbor connectivity is much lower than half of their degree \cite{D11}.
This limits possible generalizations of Theorems~\ref{th:hypercube} and \ref{th2}. On the other hand, the only abelian Cayley graphs known with larger neighbor connectivity are cycles, $\kappa_{NB}(C_n)= \delta(C_n)/2+1$ for $n\ge6$. It is therefore still possible that the upper bound provided by Theorem~\ref{thm:cayley} may be extended to all abelian Cayley graphs, thus providing a full solution to Doty's problem. 
\section*{Acknowledgments} 
This work was partially supported by the Czech Science Foundation Grant 19-08554S (Tom\'a\v{s} Dvo\v{r}\'ak), the China Postdoctoral Science Foundation Grant 2018M631322 (Mei-Mei Gu) and
OP RDE project CZ.02.2.69/0.0/0.0/16\_027/0008495 International Mobility of Researchers at Charles University (Mei-Mei Gu). 


\begin{thebibliography}{99} \small 
\baselineskip=12pt 
\parskip 0pt

\bibitem{B}
J.A.~Bondy, U.S.R.~Murty, Graph Theory, Springer, New York, 2008.

\bibitem{CoWu}
M.B.~Cozzens, S.Y.~Wu, Extreme values of the edge-neighbor-connectivity, Ars Combin. {\bf 39} (1995), 199--210.

\bibitem{D97}
K.~Day, A.E.~Ai-Ayyoub, Fault diameter of $k$-ary $n$-cube networks, IEEE Trans. Parallel Distrib. Syst. {\bf 8}:9 (1997), 903--907, \doi{10.1109/71.615436}.

\bibitem{D1}
L.L.~Doty, A new bound for neighbor-connectivity of abelian Cayley graphs, Discrete Math. {\bf 306} (2006), 1301--1316, \doi{10.1016/j.disc.2005.09.018}.

\bibitem{D11}
L.L.~Doty, Bounding neighbor-connectivity of abelian Cayley graphs, Discuss. Math. Graph Theory {\bf 31} (2011), 475--491, \doi{10.7151/dmgt.1559}.

\bibitem{D2}
L.L.~Doty, R.J.~Goldstone, C.L.~Suffel, Cayley graphs with neighbor connectivity one, SIAM J. Discrete Math. {\bf 9} (1996), 625--642, \doi{10.1137/S0895480194265751}.

\bibitem{G14}
M.-M. Gu, R.-X Hao, $3$-extra connectivity of $3$-ary $n$-cube
networks, Inform. Process. Lett. {\bf 114} (2014), 486--491, \doi{10.1016/j.ipl.2014.04.003}.

\bibitem{Gu}
G.~Gunther, B.~Hartnell, On minimizing the effects of betrayals in a resistance movement, in:
Proc. 8th Manitoba Conf. Numer. Math. Comp. (1978), pp. 285--306.

\bibitem{GH91}
G.~Gunther, B.~Hartnell, On $m$-connected and $k$-neighbour-connected graphs, in:
Proceedings of the Sixth Quadrennial International Conference on the Theory and Applications of Graphs, 1991, pp. 585--596.

\bibitem{G2}
G.~Gunther, B.~Hartnell, R. Nowakowski, Neighbor-connected graphs and projective planes, Networks {\bf 17} (1987), 241--247, \doi{10.1002/net.3230170208}.

\bibitem{H12}
S.-Y.~Hsieh, Y.-H.~Chang,  Extraconnectivity of $k$-ary $n$-cube networks, Theor. Comput. Sci. {\bf 443} (2012), 63--69, \doi{10.1016/j.tcs.2012.03.030}.

\bibitem{Liu19}
A.~Liu, S.~Wang, J.~Yuan, X.~Ma, The $h$-extra connectivity of $k$-ary $n$-cubes, Theor. Comput. Sci. {\bf 784} (2019), 21--45, \doi{10.1016/j.tcs.2019.03.030}.

\bibitem{Lovasz}
L.~Lov\'asz, Combinatorial Problems and Exercises, North Holland, Amsterdam, 1993.

\bibitem{Mulder}
H.M.~Mulder, The interval function of a graph, Mathematical Centre Tracts 132, Mathematisch Centrum, Amsterdam, 1980.

\bibitem{S2}
Y.-J.~Shang, R.-X.~Hao, M.-M.~Gu,
Neighbor connectivity of two kinds of Cayley graphs, 
Acta Math. Appl. Sin., Engl. Ser. {\bf 34} (2018), 386--397, \doi{10.1007/s10255-018-0739-9}.

\bibitem{Z1}
X.-B.~Zhao, Z.~Zhang, Q.~Ren, Edge neighbor connectivity of Cartesian product graph $G\times K_{2}$, Appl. Math. Comput.  {\bf 217} (2011), 5508--5511, \doi{10.1016/j.amc.2010.12.022}.
\end{thebibliography}
\end{document}